\documentclass{article}

\usepackage{amssymb} 
\usepackage{graphicx} 

\newtheorem{theorem}{Theorem}
\newtheorem{lemma}{Lemma}

\newtheorem{proof}{Proof}

\begin{document}
%
\title{Quantization Errors of fGn and fBm Signals}

\author{Zhiheng Li, Li Li\footnote{Corresponding author:
li-li@mail.tsinghua.edu.cn}, Yudong Chen, Yi Zhang\\
\\
Department of Automation, Tsinghua University, Beijing, China
100084}

\maketitle

\begin{abstract}
In this Letter, we show that under the assumption of high
resolution, the quantization errors of fGn and fBm signals with
uniform quantizer can be treated as uncorrelated white noises.
\end{abstract}

\section{Introduction}
\label{sec:1}

Fractional Gaussian noise (fGn) and fractional Brownian motion (fBm)
provide convenient ways to describe stochastic processes with
long-range dependencies. Thus, they have received continuing
interests in various fields and have many applications, e.g.
modeling the communication networks flow and economic times
series \cite{MandelbrotvanNess1968}-
\cite{Mishura2008}.

Of particular interest is the estimation of the Hurst exponent $H$
of a fGn or fBm process. In practice, such estimations are usually
done on the sampled and quantized time series. For example, texture
images are often viewed as 2D fBm signals uniformly quantized to the
$0-255$ scale \cite{ChenKellerCrownover1993},
\cite{Pesquet-PopescuVehel2002}. As will be shown later, the
quantization error might significantly affect the estimation result.

To the best of our knowledge, no reports discussed the effect of
quantization errors of fGn and fBm processes. In this Letter, we
will show that under the assumption of high resolution, the
quantization error can be viewed as a white noise added to the
sampled fGn or fBm signal.

\section{The Discrete-Time fGn and fBm Signals}
\label{sec:2}

There exist different kinds of discrete-time approximation for the
continuous-time fGn and fBm processes, e.g.
\cite{MandelbrotWallis1969}-
\cite{LeeZhaoNarasimhaRao2003}. In this Letter, we will use the
two-step discrete-time approximation signals defined in
\cite{Meade1993}:

1) First, the standard discrete-time fGn process $W^H (n)$ with
Hurst exponent $H \in (0,1)$ is defined as a weighted sequence of a
standard Gaussian white noise $W(n)$
\begin{equation}
\label{equ:1} W^H (n) = \left\{ \begin{array}{ll}
\sum_{i = 0}^{n} h_{n-i}^{H-\frac{1}{2}} W(i) & \textrm{for } H \ne \frac{1}{2} \\
W(n) & \textrm{for } H = \frac{1}{2}
\end{array} \right.
\end{equation}

\noindent where $W(n) \sim N(0, 1)$ and $h_n^{H-\frac{1}{2}} =
\frac{\Gamma(n + H - \frac{1}{2})}{\Gamma(H - \frac{1}{2})\Gamma(n +
1)}$, $n \in \mathbb{N} \cup \{0\}$.

2) Second, the discrete-time fBm process $B^H (n)$ is represented as
the running sum of $W^H(n)$
\begin{equation}
\label{equ:2} B^H (n) = \sum_{i=0}^{n} W^H(i) = \sum_{i=0}^{n}
\sum_{j=0}^{i} h_{i-j}^{H-\frac{1}{2}} W(j)
\end{equation}

As pointed out in \cite{GrangerJoyeaux1980}, \cite{Hosking1981},
Eq.(\ref{equ:1})-(\ref{equ:2}) are indeed ARFIMA processes, which
can be rewritten in terms of the lag operator $L$:
\begin{equation}
\label{equ:3} W^H (n) = \left\{ \begin{array}{ll}
(1-L)^{\frac{1}{2}-H} W(n) & \textrm{for } H \ne \frac{1}{2} \\
W(n) & \textrm{for } H = \frac{1}{2}
\end{array} \right.
\end{equation}

and
\begin{equation}
\label{equ:4} (1-L) B^H (n) = B^H (n)-B^H (n-1) = W^H(n)
\end{equation}

\noindent where $(1-L)^d = \sum_{k=0}^{\infty} \frac{\Gamma(k -
d)}{\Gamma(-d)\Gamma(k + 1)} L^k$. The truncated formulas in
(\ref{equ:1})-(\ref{equ:4}) are equivalent to infinite formulas with
$W(i) = 0$ for $i \le 0$; they are used here since we usually
consider the case with $W^H(i) = 0$ for $i
\le 0$ \cite{MandelbrotvanNess1968}-
\cite{Mishura2008}.

For clarity, we focus on the above standard discrete-time fBm
process but the conclusions can be easily extended to general cases.

\section{The High-resolution Quantization Errors of the fGn and fBm Signals}
\label{sec:3}

The use of high resolution theory for error process analysis can
date back to late 1940s \cite{ClavierPanterGrieg1947},
\cite{Bennett1948}. In \cite{Bennett1948}, Bennett demonstrated that
under the assumption of high resolution and smooth density of the
sampled random process, the quantization error behaves like an
additive white noise. In other words, the quantization error has
small correlation with the signal and an approximately white
spectrum; see also the good surveys in \cite{Gray1990}-
\cite{WidrowKollar2008}. In the sequel, we will show that this
conclusion also holds for fGn and fBm signals. Our proof mainly uses
the results in \cite{ChouGray1992}.

Suppose the original discrete-time fGn or fBm sequence $S^H(n)$ is
bounded within $[-b, b]$ in a finite time horizon $[0, t]$ and an
$M$-level uniform quantizer in $[-b, b]$ is applied. We also assume
the sample rate and the resolution of the quantizer are high enough.

As shown in \cite{ChouGray1992}, by defining $\Delta =
\frac{2b}{M-1}$, the normalized quantization noise $e(n)$ of
$S^H(n)$ can be represented as the normalized quantization noise of
the sigma-delta modulator for $S^H(n)$:
\begin{equation}
\label{equ:5} e(n) \triangleq \frac{1}{2} - \langle \sum_{i=0}^{n}
\left( \frac{(M-1) S^H(n)}{\Delta} + \frac{1}{2} \right) \rangle
\end{equation}

\noindent where $\langle z \rangle = x$ mod $1$ is the fractional
part of $x$.

In the following subsections, we will discuss the quantization
errors of the fGn and fBm signals, respectively.

\subsection{fGn Signals}
\label{sec:3.1}

For the sigma-delta modulator, we have the following useful lemmas.

\indent

\begin{lemma} \cite{ChouGray1992} Define an casual stable MA process $x(n)$
\begin{equation}
\label{equ:6} x(n) = \psi(L) z(n) = \sum_{i=0}^{n} \psi_i z(n-i)
\end{equation}

\noindent where $z(n)$ is an i.i.d process having a certain
distribution with smooth density. If the regression coefficients
$\psi_i$ of this MA process satisfy that there exist $\eta > 0$ and
infinitely many values of $r$ such that $\left \vert \psi_0 + ... +
\psi_r \right \vert > \eta$, then the distribution of the normalized
quantization error $e(n)$ under modulo sigma-delta modulation
converges to the uniform distribution in $[-\frac{1}{2},
\frac{1}{2}]$ under the assumption of high resolution.
\footnote{\textit{Lemma 1} is actually a slightly modified version
of \textit{Property 3} in \cite{ChouGray1992}, but it is not
difficult to see that the proof in \cite{ChouGray1992} can be
applied to \textit{Lemma 1} with little modification.}
\end{lemma}


\begin{lemma} \cite{GradshteynRyzhik2007} The following series
${_1F_0}(\alpha, z)$ is a special hypergeometric series
\begin{equation}
\label{equ:7} \sum_{i=0}^{\infty} \frac{\Gamma(\alpha +
i)}{\Gamma(\alpha)\Gamma(i+1)} z^{i} = {_1F_0}(\alpha, z)
\end{equation}
\noindent where $\alpha, z \in \mathbb{C}$.

If $1 \ge \textrm{Re}(\alpha) > 0$, the series converges throughout
the entire unit circle $\left| z \right| = 1$ except for the point
$z = 1$. If $\textrm{Re}(\alpha) < 0$, the series converges
(absolutely) throughout the entire unit circle $\left| z \right| =
1$. Particularly, the special hypergeometric series ${_1F_0}(\alpha,
1)$ converges to $0$, when $\textrm{Re}(\alpha) < 0$.
\end{lemma}

\indent

Now we can prove the first main result of this Letter using
\textit{Lemma 1} and \textit{Lemma 2}.

\indent

\begin{theorem} The quantization error of fGn process with uniform
quantizer is asymptotically uniformly distributed in $[-\frac{1}{2},
\frac{1}{2}]$ under the assumption of high resolution.
\end{theorem}

\begin{proof} We will discuss the following three cases of $0 < H <
\frac{1}{2}$, $H = \frac{1}{2}$, and $\frac{1}{2} < H <1$,
respectively. From Eq.(\ref{equ:3}), we know that the coefficients
of an fGn signal are $\psi_{G,i}=h_{i}^{H-\frac{1}{2}}$.

i) For $H \in (0, \frac{1}{2})$, $\Gamma(H-\frac{1}{2}) < 0$.
Because $0 > H-\frac{1}{2} > - \frac{1}{2}$, $\sum_{i=0}^{\infty}
\psi_{G,i} = \sum_{i=0}^{\infty} \frac{\Gamma(i +
H-\frac{1}{2})}{\Gamma(H-\frac{1}{2})\Gamma(i + 1)} = 0$ by
\textit{Lemma 2}. Thus we cannot directly apply \textit{Lemma 1}
here. However,  the convergence speed of this series satisfies
$\sum_{i=0}^{n} \frac{\Gamma(i +
H-\frac{1}{2})}{\Gamma(H-\frac{1}{2})\Gamma(i+1)} \ge
\frac{1}{\sqrt{n}}$ for $H \in (0, \frac{1}{2})$
\cite{WhittakerWatson1996}. Based on these observations, we will
derive the limit distribution of the quantization error through the
limit of its characteristic function.

As proven in \cite{ChouGray1992}, we can rewrite Eq.(\ref{equ:5}) as
\begin{equation}
\label{equ:8} e(n) = 1 - \frac{1}{2} \langle \theta(n) \rangle
\end{equation}

\noindent where $\theta(n) \triangleq \sum_{i=0}^{n} \left(
\frac{W^H(i)}{\Delta} + \frac{1}{2} \right)$.

The corresponding characteristic function can be written as
\begin{eqnarray}
\label{equ:9}
& \left \vert \Phi_{\langle \theta(n) \rangle} (2 \pi l) \right \vert \nonumber \\
& = \left \vert \textrm{E} \left\{ \exp \left[ 2 \pi
\frac{l}{\Delta}
\left( W^H(n) + ... + W^H(0) \right) \right] \right\} \right \vert \nonumber \\
& = \left \vert \textrm{E} \left\{ \exp \left[ 2 \pi
\frac{l}{\Delta}
\left( \sum_{i=0}^n h_{i}^{H-\frac{1}{2}} W(n-i) + ... + W(0) \right) \right] \right\} \right \vert \nonumber \\
\end{eqnarray}

The innermost sum in Eq.(\ref{equ:9}) can be grouped as
\begin{eqnarray}
\label{equ:10} & \sum_{i=0}^n h_i^{H-1/2} W(n-i) + ... + W(0) \nonumber \\
& = h_{0}^{H-\frac{1}{2}} W(n) + (h_{0}^{H-\frac{1}{2}} +
h_{1}^{H-\frac{1}{2}}) W(n-1) +
\nonumber \\
& ... + (h_{0}^{H-\frac{1}{2}} + ... + h_{n}^{H-\frac{1}{2}}) W(0)
\end{eqnarray}

Hence
\begin{eqnarray}
\label{equ:11} & & \lim_{n \rightarrow \infty} \Phi_{\langle
\theta_n \rangle} (2 \pi l) \nonumber \\
& = & \lim_{n \rightarrow \infty} E \left\{ \prod_{i=0}^{n} \Phi_W
\left( 2 \pi \frac{l}{\Delta} \sum_{j=0}^{i} h_{j}^{H-\frac{1}{2}}
 \right) \right\}
\end{eqnarray}

Notice that the characteristic function of a standard Gaussian
process is $\Phi_W ( \omega ) = \exp(- \frac{1}{2} \omega^2)$,
$\omega \in \mathbb{R}$. We have
\begin{eqnarray}
\label{equ:12} & & \left| \Phi_W \left( 2 \pi \frac{l}{\Delta}
\sum_{j=0}^{i}
h_{j}^{H-\frac{1}{2}} \right) \right| \nonumber \\
& = & \left| \Phi_W \left( 2 \pi \frac{l}{\Delta} \sum_{j=0}^{i}
\frac{\Gamma(H-\frac{1}{2} + j)}{\Gamma(H-\frac{1}{2})\Gamma(j+1)}
\right) \right|
\nonumber \\
& \le & \left| \Phi_W \left( 2 \pi \frac{l}{\Delta}
\frac{1}{\sqrt{i}} \right) \right|
\end{eqnarray}

The harmonic series $\sum_{i=1}^{\infty} \frac{1}{i}$ diverges; in
other words, for any small positive number $\epsilon >
0$, 
we can always find a large enough integer $n^*$ such that
\begin{eqnarray}
\label{equ:13} \sum_{i=1}^{n^*} \frac{1}{i} > - \ln
\left(\epsilon\right) \left( \frac{\Delta}{2 \pi l} \right)^2
\end{eqnarray}

\noindent or equivalently
\begin{eqnarray}
\label{equ:14} \exp \left[ - \left( \frac{2 \pi l}{\Delta} \right)^2
\sum_{i=1}^{n^*} \frac{1}{i} \right] < \epsilon
\end{eqnarray}

From (\ref{equ:11}), (\ref{equ:12}) and (\ref{equ:14}), it follows
that for any small $\epsilon >0$, there exists a large enough
integer $n^*$ such that for all $n > n^*$,
\begin{eqnarray}
\label{equ:15} \left| \Phi_{\langle \theta_n \rangle} (2 \pi l)
\right| & \le & \left| E \left\{ \prod_{i=0}^{n-1} \Phi_W \left( 2
\pi \frac{l}{\Delta} \frac{1}{\sqrt{i}}
 \right) \right\} \right| \nonumber \\
& = & \left| E \left\{ \exp \left[ - \left( 2 \pi \frac{l}{\Delta}
\right)^2 \sum_{i=0}^{n} \frac{1}{i} \right] \right\} \right| \nonumber \\
& < & \epsilon
\end{eqnarray}

\noindent which means
\begin{eqnarray}
\label{equ:16} \lim_{n \rightarrow \infty} \Phi_{\langle \theta_n
\rangle} (2 \pi l) = \left\{ \begin{array}{ll}
1 &, l = 0 \\
0 &, l \ne 0
\end{array} \right.
\end{eqnarray}

Therefore, the distribution of $\langle \theta(n) \rangle$ converges
to the uniform distribution in $[0, 1]$ and the limit distribution
of $e(n)$ is $U[-\frac{1}{2}, \frac{1}{2}]$.

ii) For $H = \frac{1}{2}$, we directly have $\sum_{i=0}^{\infty}
\psi_{G,i} = 1 \ne 0$, so \textit{Lemma 1} applies.

iii) For $H \in (\frac{1}{2}, 1)$, $\sum_{i=0}^{\infty} \psi_{G,i} =
\sum_{i=0}^{\infty} \frac{\Gamma(i +
H-\frac{1}{2})}{\Gamma(H-\frac{1}{2})\Gamma(i + 1)}$ diverges by
\textit{Lemma 2}. It then follows from the definition of divergence
that the premise of \textit{Lemma 1} are satisfied.

\end{proof}

\subsection{fBm Signals}
\label{sec:3.2}

To analyze the quantization error of fBm processes, we need another
lemma from \cite{ChouGray1992}.

\indent

\begin{lemma} \cite{ChouGray1992} Define an AR(1) process $x(n)$ as
\begin{equation}
\label{equ:17} x(n) - x(n-1) = z(n)
\end{equation}

If the input $z(n)$ is a stationary independent increments, the
normalized quantization noise $e(n)$ of the modulo sigma-delta
modulation converges to the uniform distribution in $[-\frac{1}{2},
\frac{1}{2}]$ and has a white spectrum under assumption of high
resolution.
\end{lemma}

\indent

\textit{Lemma 3} directly applies to the quantization error of fBm
processes with $H = \frac{1}{2}$ (indeed, the Brownian motion). For
$H \ne 1/2$ the techniques used for proving this lemma in
\cite{ChouGray1992} can also be adopted to characterize the
quantization error.

\indent

\begin{theorem} The quantization noise with uniform quantizer of
fBm process is asymptotically uniformly distributed and white under
the assumption of high resolution.
\end{theorem}

\begin{proof}

i) For $H = \frac{1}{2}$, we have
\begin{equation}
\label{equ:18} B^H (n) - B^H (n-1) = W^H(n)
\end{equation}

\noindent which follows directly from \textit{Lemma 3}.

ii) For $H \in (0, \frac{1}{2}) \cup (\frac{1}{2}, 1)$, we will
derive the limit distribution using the characteristic function. We
can define
\begin{equation}
\label{equ:19} e(n) = 1 - \frac{1}{2} \langle \delta(n) \rangle
\end{equation}

\noindent where $\delta(n) \triangleq \sum_{i=0}^{n} \left(
\frac{B^H(i)}{\Delta} + \frac{1}{2} \right)$.

From Eq.(\ref{equ:1}) and Eq.(\ref{equ:2}), we know the regression
coefficients of an fBm signal are $\psi_{B,i} = \sum_{j=0}^i (i-j)
h_j^{H-\frac{1}{2}}$. Hence in the case of $H \in (\frac{1}{2}, 1)$,
there exists $\eta > 0$ and infinitely many values of $r$ such that
\begin{eqnarray}
\label{equ:20} & \left \vert \psi_{B,0} + \cdots + \psi_{B,r} \right
\vert = \left \vert rh_0^{H-\frac{1}{2}} + \cdots
h_r^{H-\frac{1}{2}}\right \vert \nonumber \\
& > \left \vert h_0^{H-\frac{1}{2}} + \cdots
h_r^{H-\frac{1}{2}}\right \vert > \eta
\end{eqnarray}

\noindent where the last inequality follows from part iii) of the
proof of \textit{Theorem 1}. Now \textit{Lemma 1} applies and the
distribution of $e(n)$ converges to $U[-\frac{1}{2}, \frac{1}{2}]$.

Similarly, we can prove the case of $H \in (0, \frac{1}{2})$.
Because $h_0^{H-\frac{1}{2}} > 0$, $h_i^{H-\frac{1}{2}} < 0$ for $i
> 0$, we have
\begin{eqnarray}
\label{equ:21} r h_{0}^{H-\frac{1}{2}} + \cdots +
h_{r}^{H-\frac{1}{2}} > 0
\end{eqnarray}

\noindent for $r \in \mathbb{N}$ and thus
\begin{eqnarray}
\label{equ:22} & \lim_{r->\infty} | \psi_{B,0} + ... + \psi_{B,r}|
\nonumber \\
& = \lim_{r \rightarrow \infty} \left \vert
r h_{0}^{H-\frac{1}{2}} + (r-1) h_{1}^{H-\frac{1}{2}} + \cdots + h_{r}^{H-\frac{1}{2}} \right \vert \nonumber \\
& > \lim_{r \rightarrow \infty} \left \vert h_{0}^{H-\frac{1}{2}} +
... + h_{r}^{H-\frac{1}{2}} \right \vert = 0
\end{eqnarray}

Again \textit{Lemma 1} applies.

\indent

From Eq.(\ref{equ:19}), the limit correlation
$\bar{\textrm{R}}\left((e(n), e(n+k)\right)$ can be written as
\begin{eqnarray}
\label{equ:23} & \bar{\textrm{R}} \left((e(n), e(n+k)\right) \nonumber \\
& = \frac{1}{4} - \bar{\textrm{E}} \left( \langle \delta(n) \rangle
\right) + \bar{\textrm{E}} \left \{ \langle \delta(n) \rangle
\langle \delta(n+k) \rangle \right \}
\end{eqnarray}

\noindent where the limit mean $\bar{\textrm{E}} \left( x(n) \right)
= \lim_{N \rightarrow \infty} \frac{1}{N} \sum_{n=1}^N \textrm{E}
\left( x(n) \right)$.

For $k \ne 0$, $( \langle \delta(n) \rangle, \langle \delta(n+k)
\rangle)$ converges in distribution to a random variable which is
uniformly distributed in $[0,1) \times [0, 1)$; thus we have
$\bar{\textrm{E}} \left( \langle \delta(n) \rangle, \langle
\delta(n+k) \right) = \int_0^1 \int_0^1 u v du dv = \frac{1}{4}$.

For $k = 0$, we have $\bar{\textrm{E}} \left( \langle \delta(n)
\rangle^2 \right)  = \frac{1}{3}$ and $\bar{\textrm{R}} \left(
\langle \delta(n) \langle \delta(n) \right) = \frac{1}{3}$.

By (\ref{equ:23}), we have
\begin{equation}
\label{equ:24} \bar{\textrm{R}} \left((e(n), e(n+k)\right) = \left\{
\begin{array}{ll}
\frac{1}{12} &, k = 0 \\
0 &, k \ne 0
\end{array} \right.
\end{equation}

Thus, the normalized quantization error is white and is
asymptotically uncorrelated with the output of the quantized signal.
Therefore, we prove the whole conclusion.

\end{proof}

\section{Some Simulation Results}
\label{sec:4}

Fig.1 shows a typical Power Spectral Density (PSD) for the
quantization error of a 1D quantized fBm signal, which indicates
that the normalized quantization error is indeed white.

\begin{figure}[h]
\label{fig:1} \centering
\includegraphics[width=2.5in]{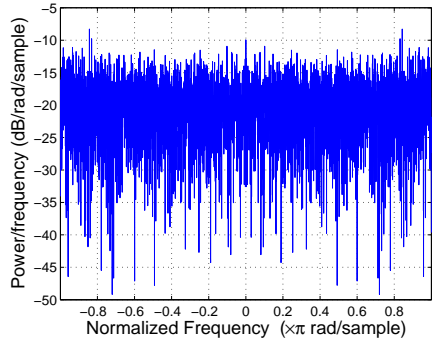}
\caption{The PSD estimate for the quantization error of a 1D
quantized fBm signal, where $H = 0.2$, the quantization scale is
$\Delta = 1$. The PSD is estimated via periodogram method.}
\end{figure}

Fig.2 shows the PCA eigen-spectrum of the original fBm signal, the
quantized fBm signal and the quantization error. According to the
results given in
\cite{BaykutErolAkgul2006}-\cite{LiHuChenZhang2009}, when the
sampling data length $K$ is a sufficiently large constant, the PCA
eigenvalue spectrum of the auto-correlation of a 1D fBm process with
Hurst exponent $H$ decays as a power-law
\begin{equation}
\label{equ:25} \tilde{\lambda}_k \sim k^{-(2H+1)}, \indent k=1, ...,
K
\end{equation}

It is also proven in
\cite{BaykutErolAkgul2006}-\cite{LiHuChenZhang2009} that the
numerical eigen-spectrum of a white noise should be a straight line
with slope $\alpha_0 \approx 0$ in the log-log scale (the slope is
not strictly $0$ because of the finite sampling length effect).
Moreover, when the 1D fBm signal is corrupted with additive white
noise and the SNR is large enough, the eigenvalue spectrum of the
corrupted signal crossovers from $\alpha_1 = -(2H+1)$ to $\alpha_0$.
Comparing Fig.2 to the simulation results provided in
\cite{BaykutErolAkgul2006}-\cite{LiHuChenZhang2009}, we can see that
under high resolution, the quantization error behaves exactly like a
certain additive white noise.

\begin{figure}[h]
\label{fig:2} \centering
\includegraphics[width=2.5in]{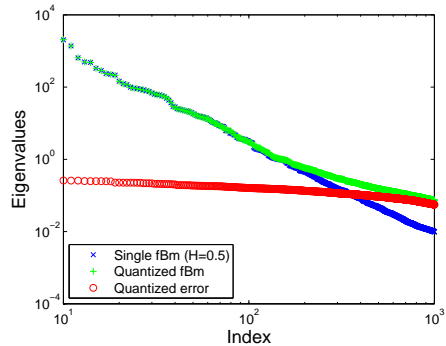}
\caption{PCA eigen-spectrum of the auto-correlation of a standard
fBm signal, its corresponding quantized signal, and the quantization
error, where $H = 0.8$, and the quantization scale is $\Delta = 1$.}
\end{figure}


\end{document}